\documentclass[reprint,pra,showpacs]{revtex4}
\usepackage{amsmath,amssymb,amsfonts,amsthm}
\usepackage{graphicx,verbatim,booktabs}
\usepackage[percent]{overpic}

\newcommand{\steps}{N}				%Num timesteps
\newcommand{\ctgtipol}[2]{J^{#2}_{#1}}	%cost-to-go
\newcommand{\tcf}{\ctgtipol{\steps}{*}}			%Terminal cost function
\newcommand{\tcfm}[1]{J^{*#1}_\steps}	%Terminal cost function, specified
\newcommand{\st}{s}					%A general state
\newcommand{\stm}[1]{\st_{#1}}			%A state, specified
\newcommand{\stt}{\stm{T}}				%target state
\newcommand{\stf}{\stm{\steps}}			%terminal state s_N
\newcommand{\sti}{\stm{0}}				%initial state
\newcommand{\stc}{\hat{\st}}			%controlled state
\newcommand{\stcm}[1]{\hat{\stm{#1}}}
\newcommand{\str}{\st'}				%result state
\newcommand{\sr}{x}					%rectangle state
\newcommand{\srt}{x_T}				%target state rect.
\newcommand{\src}{y}				%control to state rect.
\newcommand{\SOS}{\mathcal{S}}		%Set of states
\newcommand{\ti}{i}					%A time
\newcommand{\tia}{k}				%another time
\newcommand{\SOT}{I}				%set of times
\newcommand{\SOA}{A(\st,\ti)}		%set of actions
\newcommand{\costonly}{g}
\newcommand{\cost}{\costonly(\st,\stc)}		%cost function
\newcommand{\costm}[1]{\costonly_{#1}(\st,\stc)}
\newcommand{\costsm}[1]{\costonly(\stm{#1},\stcm{#1})}
\newcommand{\polf}{\pi}				%policy function
\newcommand{\pol}{\polf(\st,\ti)}	%policy function with arguments
\newcommand{\polo}{\polf^*}			%optimal policy
\newcommand{\outcome}{\omega}
\newcommand{\bellmanit}[1]{\mathcal{U}[#1]}	%bellman iteration
\newcommand{\bellmanitpow}[2]{\mathcal{U}^{#2}[#1]}  %repeated bellman iteration
\newcommand{\bellmanitapp}[1]{\tilde{\mathcal{U}}[#1]} %approximate bellman iteration
\newcommand{\bloch}{\mathcal{B}}	%bloch sphere set
\newcommand{\Exp}{\mathsf{E}} %expectation
\newcommand{\ee}{\mathsf{e}} % exponential
\newcommand{\ii}{\mathsf{i}} % imaginary i
\newcommand{\R}{\mathbb{R}}

\newcommand{\bra}[1]{\langle #1|}
\newcommand{\ket}[1]{|#1\rangle}
\newcommand{\braket}[2]{\langle #1|#2\rangle}

\newcommand{\dens}[1]{\ket{\psi_{#1}}\bra{\psi_{#1}}}
\newcommand{\adj}[1]{{#1}^{\dag}} %adjoint
\newcommand{\noadj}[1]{{#1}^{\vphantom\dag}}
\newcommand{\tr}{\operatorname{{\mathrm tr}}}

\newcommand{\norm}[1]{||#1||}

\theoremstyle{plain}

\newtheorem{proposition}{Proposition}

\theoremstyle{definition}

\theoremstyle{remark}

\theoremstyle{plain}

\begin{document}
\bibliographystyle{plain}
\title{Discrete Quantum Control - State Preparation}
\author{Jon R. Grice}
\author{David A. Meyer}
\affiliation{Department of Mathematics\\ University of California, San Diego}
\email{jgrice@math.ucsd.edu}
\email{dmeyer@math.ucsd.edu}
\date{April 28, 2012}
\pacs{03.67.Ac, 02.30.Yy}
\begin{abstract}
A discrete-time method for solving problems in optimal quantum control is presented. Controlling the time discretized Markovian dynamics of a quantum system can be reduced to a Markov-decision process. We demonstrate this method in this
with a class of simple one qubit systems, which are also discretized in space. For the task of state preparation we
solve the examples both numerically and analytically with dynamic programming techniques.
\end{abstract}
\maketitle

%%%%%%%%%%INTRO
\section{Introduction}
	Quantum control is a discipline that will see its importance match any of the technologies that it enables --
	quantum information sciences, to take a prominent example.

	But, the problems which arise in Quantum control can be very difficult to solve. 
	This paper explores a possible direction towards formulating and solving problems from optimal
	quantum control by taking inspiration from quantum computing.
	
	The control of quantum systems has been studied for decades, see e.g. \cite{belavkin83}. 
	Because real systems and measurements are often continuous \cite{mabuchi2005},
	the usual technique in optimal quantum control is to
	study the stochastic master equation for the system being controlled. A sizable body of literature
	exists which details the control of continuous time quantum systems \cite{bvhj}.
	
	Some systems are inherently time-discrete like the Stern-Gerlach experiment or a quantum circuit,
	but any time-continuous system must be discretized if it is to be simulated on a quantum computer \cite{feynman82}.
	These time-discrete simulations are what we seek to control, in lieu of the continuous model of the system. 
	
	We model the system by slicing time into steps where the evolution on each step
	is given by a quantum circuit (including nonunitary quantum operations and measurements). 
	Each timestep's circuit effects a Markovian step on the system's state.
	Since we want to control a Markov chain we turn to the theory of Markov decision processes (MDPs) \cite{suttonbarto},
	specifically to the field of control theory called Discrete-Time Optimal Stochastic Control \cite{bertsekas}.
	
	A MDP consists of two parts, the agent and the system. The agent consists of the experimenter and a (classical) computer system, the controller. 
	The system is an open quantum system or a model thereof,
	interacting through its environment, from which a measurement apparatus takes measurements which are reported to the
	agent. The agent in turn can vary the Hamiltonian (apply controls) which acts on the system to some specified degree. The closed
	loop process of taking measurements and applying controls is repeated for either a certain amount of time or
	indefinitely, depending on the goal the experimenter has in mind. The controller is programmed to use the information
	from the measurements along with the knowledge of the system dynamics to apply the appropriate controls in order to
	best further the goals of the experimenter. In this paper we discuss methods for finding the best program, on average,
	that the controller should use.
	
The MDPs that arise in this setting have a continuous state space and transition matrices which arise from
	the dynamics of the system which give them more structure than most MDPs that are studied. It is this
	structure that gives us more power to solve these systems, for example, we can solve some MDPs analytically.

	We demonstrate solving a system by using two sample models, the affine-quadratic and the threshold models,
	so-named for the cost functions that determine their behavior.
These models have
	a simple enough state space (the interval $[0,1]$) that there is little difficulty 
	in discretizing that space and solving the Bellman equations that arise numerically. 
	
	The discrete technique presented in this paper is also applicable to more complicated models. A subsequent paper will study open systems undergoing decoherence
	and bipartite systems undergoing disentanglement.

	The two example models are chosen to be as elementary as possible,
	the system being that of a qubit undergoing a stochastic evolution in which the state remains pure.
	There is a desired target state that must be reached after a fixed amount of time, but the target state
	is unstable. 

	We solve the models first numerically and note peculiar behavior of the solutions, such as
	the as-if piecewise linear shape of the decision function graph,and corresponding parabolic shape of
	the value functions, which do not necessarily have a minimum at the target state. This means it is not
	good to stay near the target state if it is unstable and the end of the experiment is far off.
	As the value functions are propagated backwards in time as is done in dynamic programming fast convergence is evident for the value functions tending to a steady-state
	solution. This means that for a very long experiment relative to the probability of decaying away from the
	target state, the controller takes little action until the latter part of the experiment.
	
	Our parameterization of the measurement operators allows us to study the best measurements to make
	in order to control the system as efficiently as possible. Interestingly we find that having numerous
	(more than two) measurement operators can make the system easier to control, in the sense that the value functions
	are smaller compared to a system with only two measurement outcomes.
	
	These behaviors of the affine-quadratic model that are observed from numerical simulations can be explained
	by the analytic solution of the Bellman equation. Due to the simplicity of the system the Bellman equation can
	be approximately solved with known error bounds.
	\subsection{The state preparation problem}\label{distmeas}
		Algorithms in quantum computation often require an initial state to be prepared. The state preparation
		module would fill this need by preparing the desired state $\ket{\psi_T}$ (as closely as possible) at the correct time
		by starting with an arbitrary unprepared state \cite{jacobs06}. Oftentimes one can perform a
		projective measurement in which one of the outcomes is the desired state to be prepared. This 
		however, may not happen often enough, fast enough \cite{jacobs03}, or there may not be the physical resources
		available to make the correct state with a single measurement.
		
		In our two models the system begins in a known, but arbitrary pure state $\ket{\psi_0}$. After
		$N$ timesteps, the final state $\ket{\psi_N}$ is evaluated against the target state $\ket{\psi_T}$ (there
		can be more than one target state in general).
		For the cost function formulation of the finite horizon MDP (see below), the success of the state
		preparation process of measured by the terminal cost function $\tcf$. Clearly $\tcf$ should take its minima on
		the target states.
	
		For the two models described in this paper, there is a single target state $\ket{\psi_T}=\ket{0}$.
		The first model we deal with, the affine model, we choose
		\begin{equation}\label{eq:vnaff}
		\tcfm{\text{aff}}(\ket{\psi_N})=D_\text{tr}(\ket{\psi_N},\ket{\psi_T})=1-\braket{\psi_N}{\psi_T}^2,
		\end{equation}
		where $D_\text{tr}$ is the trace distance.
		Suppose a quantum algorithm expects $\ket{\psi_T}$ as an input but gets $\ket{\psi_N}$ instead.
		Then that algorithm performs worst when $\tcfm{\text{aff}}$ is greatest, in the following sense.
		Suppose the quantum algorithm is represented by a unitary
		operation $W$ applied to $\ket{\psi_T}$, followed by a POVM with operators $\{\Pi_i\}_i$.
		Outcome $i$ is observed with probability $\tr{\Pi_i\adj{W}\dens{T}W}$. This provides the
		probability distribution $\{\tr \Pi_i \adj{W}\dens{T}W \}_i$. Taking the $l^1$ (classical trace) distance between
		that distribution and the one obtained by using $\ket{\psi_N}$ as the input state,
		$\{\tr \Pi_i \adj{W}\dens{N}W\}_i$, we can get a reasonable measure of success.
		Since we are not specifying the particular algorithm we do not know the POVM that will be applied.
		However the quantum trace distance $D_\text{tr}(\ket{\psi_N},\ket{\psi_T})$
		provides an upper bound for the $l^1$ distance for any POVM\cite{NC}. 
		Since we want the measure of success to be
		smaller when the probability distributions are similar, we get equation \eqref{eq:vnaff}. 
		Equation \eqref{eq:vnaff} varies with the overlap (squared fidelity) of $\ket{\psi_N}$ and $\ket{\psi_T}$, the
		probability of observing one when making a measurement looking for the other.

		Open quantum systems can be approximated so that their environment ``forgets" its interaction with
		the system very quickly (although in fact our simple model is closed). Thus Markov processes are used
		to model quantum systems and in our discrete case these will be Markov chains.	
	\subsection{Markov decision processes}
		A Markov decision process is a controllable Markov process. We will present MDPs in a slightly different
		way than is usual, to make it easier to adapt the quantum model to the formalism of the MDP. All Markov
		processes in this paper will be Markov chains, such that the time set $\SOT=\{0, ..., N-1\}$ is finite.
		A MDP is a set of states $\SOS$ together with sets of actions $\SOA \subseteq\SOS$ that can be
		taken at each state $\st\in\SOS$ and time $\ti\in\SOT$ and transition probabilities 
		$P_{\stc,\str}$.
		If the current state at time $\ti$ is $\st$ then $P_{\stc,\str}$ is the
		probability that state $\stc\in\SOA$ transitions to state $\str$ in the next timestep where
		the action $\stc$ is an intermediary state that is chosen by the agent out of the set $\SOA$. 
		The agent chooses the action by following a policy,
		a function $\pi:\SOS\times\SOT \rightarrow \SOS $ such that $\pol \in \SOA$. (There
		are formalizations of the policy function which enable the agent to act probabilistically and in a non-Markovian 
		way but this level of generality is not needed in this paper \cite{vanderwal}.)
		
		The model allows only one way to take the state $\st$ to $\stc$ at time $\ti$ and it is assigned a control cost
		$\costm{\ti}$. An alternate formulation of MDPs uses rewards (negative of costs) associated to the
		state resulting at the next timestep. If $\costm{\ti}$ is independent of time the subscript $\ti$ is dropped.
		
		In summary, a single timestep proceeds as follows:
		\begin{equation} \label{proc:onestep}
			\stm{\ti} \stackrel{\pi_i}\longrightarrow  \stcm{\ti} \stackrel{\outcome_\ti}\longrightarrow \stm{\ti+1}
		\end{equation}
		where $\outcome_\ti$ is a random outcome choosing $\str_\ti=\stm{\ti+1}$ with probability $P_{\stcm{\ti},\str_\ti}$.
		
		Markov decision processes can be run for a definite or indefinite amount of time. 
		With the state preparation problem, the experiment
		runs for a fixed amount of time ($\steps$ timesteps). This is called the finite-horizon problem. Experiments which run for an
		indefinite amount of time, which can stop at a time unknown to the experimenter or run for an infinite amount
		of time ($\SOT=\mathbb{N}$) are interesting and arise in Quantum control but the theory of modeling such systems with MDPs will not be
		discussed in this paper.

		For a finite-horizon problem, the state of the system at the end of the experiment $\stf$ is of importance.
		The terminal cost function $\tcf(\stf)$ evaluates $\stf$. If the experiment starts with initial state $\sti$,
		with transition outcomes $\outcome$ and
		the controller follows the policy $\polf$ then the total cost is
		\begin{equation}
			G(\sti,\polf,\outcome) = \sum_{\ti=0}^{\steps-1}\costsm{\ti} + \tcf(\stf).
		\end{equation}
		where the $\stm{\ti}$ are obtained from the procedure \eqref{proc:onestep}.
		Let
		\begin{equation}
			\ctgtipol{0}{\polf}(\sti) = \Exp[G(\sti,\polf,\outcome)],
		\end{equation}
		the expected total cost of the experiment, following policy $\polf$ with initial state $\sti$. For a fixed
		$\sti$ we would like to find an optimal policy $\polo$ that minimizes $\ctgtipol{0}{\polf}(\sti)$. In that case we
		write $\ctgtipol{0}{*}(\sti)=\ctgtipol{0}{\polo}(\sti)$.

		The Dynamic programming algorithm \cite{bellman,bertsekas,suttonbarto} is a well-known technique for finding the
		optimal policies $\polo$. To state it, we define the optimal cost-to-go at time $\tia$,
		\begin{equation}
			\ctgtipol{\tia}{*}(\st) = \Exp\left[\sum_{\ti=\tia}^{\steps-1}\costsm{\ti}+\tcf(\stf)\right],
		\end{equation}
		where, as before, the $\stm{\ti}$ are obtained from the process \eqref{proc:onestep}. The optimal cost-to-go
		tells us the minimum remaining total cost when starting in state $\st$ at time $\ti$.
		Under suitable conditions \cite{bertsekas},
		\begin{equation}\label{dpa}
			\ctgtipol{\ti}{*}(\st) = \bellmanit{\ctgtipol{\ti+1}{*}(\st)} = 
									\inf_{\stc \in \SOA}\left\{\Exp\left[\cost+\ctgtipol{\ti+1}{*}(\str)\right] \right\},
		\end{equation}
		for $0 \le \ti < \steps$. The dynamic programming algorithm \eqref{dpa} states that the optimal
		cost-to-go functions can be computed by iterating backwards in time.
		
		We will be concerned with cases when the infimum in equation \eqref{dpa} is
		achieved and so (an) optimal policy $\polo(\st,\ti)=\stc$ is obtained from taking an argmin in equation \eqref{dpa}.
		In the next section we will apply this framework to a general discretized physical system.
	\subsection{The state preparation problem as a Markov decision process}\label{mdp}
		By truncating the Baker-Campbell-Hausdorff (BCH) formula or using the Trotter formula 
		\cite{NC,feynman82,sornborger,bvh} on the system dynamics of
		a finite-dimensional open quantum system, a discrete time approximation can be obtained for the system:
		\begin{equation}\label{dia:circuit}
			\stm{\ti} \longrightarrow \framebox{Controllable $U_c$} \longrightarrow
								  \framebox{Free $U_f$} \longrightarrow
								  \framebox{Meas.} \longrightarrow
								\framebox{Noise} \longrightarrow
								\stm{\ti+1}.
		\end{equation}
		The circuit \eqref{dia:circuit} represents a single timeslice of system evolution. It is applied repeatedly to 
		the starting state $\sti$. The ordering of the blocks in \eqref{dia:circuit} becomes less relevant as the length 
		of the timeslice $\Delta t \rightarrow 0$. Questions arise pertaining to the behavior of the discrete model
		in the continuous time limit. These issues are discussed in \cite{bvh}.
		
		The structure of the circuit mirrors our formulation of a MDP. The state $\stm{\ti}$ enters from the left.
		The $U_c$ block is the part where the agent takes an action, choosing a new state $\stcm{\ti}\in\SOA$. The
		action is the state resulting from a unitary operation applied to $\st$ chosen from a set of control Hamiltonians
		$\mathcal{H}_\ti$ that the agent has available at time $\ti$. 
		Specifically, $\SOA=\{\ee^{\ii H\Delta t}s \mid H \in \mathcal{H}_\ti\}$.
		
		If $\SOA$ does not depend on $\st$ and $\ti$ then we can assume that
		any desired action is in $\SOA$, and disallow physically unrealistic or undesirable actions at a given state by
		making the cost prohibitive. In that case (and in the examples) we effectively set $\SOA=\SOS$. If the action
		is $\stc$, the cost paid at timestep $\ti$ is $\cost$. The cost function
		will be explained in a later section. If there are multiple Hamiltonians in $\mathcal{H}_\ti$ which result in
		the same action we choose the best one under the same physical constraints that the cost function was chosen
		(in order to make the cost function well-defined).
		
		The free evolution unitary $U_f$ is the next step. This is coherent evolution of the system that
		cannot be controlled directly by the controller. 
		This step can be folded into the previous step by composing $U_f$ with the action.
		
		The next two steps, measurement and noise can also be composed together into a single
		non-trace preserving quantum operation that sends the state $\stc$ to one of the
		states $\{R_{\alpha}(\stc)\}_{\alpha \in \Omega} \subseteq \SOS$ with probability distribution
		$\Pr(\alpha|\stc)=P_{\stc,R_{\alpha}(\str)}$ and set of measurement outcomes $\Omega$. In the finite dimensional
		system we consider, $\Omega$ will be a finite set. Then the Dynamic programming algorithm \eqref{dpa}
		becomes
		\begin{equation}\label{dpa2}
			\ctgtipol{\ti}{*}(\st) = 
			\min_{\stc \in \SOA}\left\{\int_{\Omega}\cost+\ctgtipol{\ti+1}{*}(R_\alpha(\stc))\ \mathrm{d}\Pr(\alpha|\stc) \right\},
		\end{equation}
		A few properties of the Bellman iteration $\bellmanit{\cdot}$ from \eqref{dpa} that are useful for analyzing the
		examples will be summarized.
		\begin{proposition}\label{propsu}
		Let $\SOS$ be compact and $\Omega$ be finite. Then
		$\bellmanit{\cdot}$ is a map on $\mathcal{C}(\SOS)$. That is, $\bellmanit{J}$ is continuous if $J$ is. Additionally,
		\begin{enumerate}
			\item If $\cost \ge 0$ for all $\st,\stc \in \SOS$ with equality for $\st=\stc$ then the constant function
					$J(s)\equiv C$ is a fixed point for
					$\bellmanit{\cdot}$ for all $C$.
			\item $\norm{\bellmanit{J}-\bellmanit{V}}_{\infty}		\label{contmap}
					\le \norm{J-V}_{\infty}$ In particular, $\bellmanit{\cdot}$ is a continuous map.
		\end{enumerate}
		\end{proposition}
		\begin{proof}
		Item \ref{contmap}: Fix $\st$. Then 
		\begin{align*}
			|\bellmanit{J}(\st)-\bellmanit{V}(\st)| & = |\min_{\stc}\{\cost+\sum J(R_\alpha(\stc))\Pr(\alpha|\stc)  \} \\
					& \qquad - \min_{\stc}\{\cost+\sum V(R_\alpha(\stc))\Pr(\alpha|\stc) \}| \\
				    \le & \max_{\stc}|\sum (J(R_\alpha(\stc))-V(R_\alpha(\stc)))\Pr(\alpha|\stc)|  \\
				    \le & \max_{z}|J(z)-V(z)|\sum \Pr(\alpha|\stc) \\
				    =   & \norm{J-V}_{\infty}. 
			\end{align*}
		\end{proof}
		We expect proposition \ref{propsu} is well-known in the field of discrete-time stochastic control
		but are unable to find sources. We now turn our attention to the specific models.
		%%%%MODELS
		%3 Our model specifics
\section{Optimal Control Problems for Cylindrically Symmetric Qubit States}
    We will demonstrate this technique with two single qubit models. In order to make the optimization procedure needed
    to solve the Bellman equation straightforward, we choose one of the simplest one-qubit systems, one that 
    can be parameterized by a single real number. One way to satisfy this constraint is to confine an arbitrary state
    $\st$ in the system $\SOS$ to half of a great circle on the surface of the Bloch sphere $\bloch$.
    We choose a basis $\{\ket{0},\ket{1}\}$ so that if $\stm{\sr} \in \SOS$, then
    \begin{equation}
		\stm{\sr}=\sqrt{1-\sr}\ket{0}+\sqrt{\sr} \ket{1},
	\end{equation}
    for $\sr \in [0,1]$. 
	We will hereafter abuse the notation by identifying the system $\SOS$ with its parameterization $[0,1]$ 
	and state $\stm{\sr}$ with its parameter $\sr$.
    When an additional parameter $\phi$ is introduced to represent relative phase
	the entire Bloch sphere is swept out by rotating the set $\SOS$ about the $z$-axis (the axis
	containing $\ket{0}$ and $\ket{1}$) as $\phi$
	is varied. We may therefore consider any state  $\st \in \SOS$ to be the equivalence class
	$\{\mathcal{R}_z(\phi)\st\}$ ($\mathcal{R}_z$ is the rotation about the $z$-axis operator) making 
	the set of equivalence classes $\SOS$ a cylindrically symmetric
	version of $\bloch$.
	
	There are several reasons why our cylindrically-symmetric $\SOS$ is not an unrealistic choice.
	If the experiment requires the preparation of a unpopulated state $\ket{0}$ without any consideration
	to the relative phase, then the terminal cost function will be cylindrically symmetric. If the control
	step is strobing an appropriately aligned magnetic field, (for example, one with no $\sigma_x$ or $\sigma_z$
	in the Hamiltonian) then the control cost function $\costonly$ is cylindrically symmetric. In fact, there
	are several families of physically-motivated operators which preserve $\SOS$ and hence are cylindrically
	symmetric.
	
\subsection{Operators on $\SOS$}
	We can think of a noisy quantum operation $\mathcal{E}:\st\mapsto \str$ being applied to a system 
	$\rho_\st=\ket{\psi_\st}\bra{\psi_\st}$ as the system interacting 
	with its environment $\rho_\text{e}$. 
	This can be modeled by an entangling unitary operator $U_\text{se}$ acting on the 	
	environment-system $\rho_\text{se}$ followed by a measurement $\mathcal{M}_\text{e}$ on the environment. 
	When the record of
	$\mathcal{M}_\text{e}$ is ignored the resultant system state 
	$\rho_{\str}=\tr_\text{e}(\sum_j \noadj{M_j}\noadj{U_\text{se}}\rho_\st \adj{U_\text{se}}\adj{M_j})$ 	
	is left as a mixed state meaning that in general $\str \not\in \SOS$. Now suppose the experimenter is allowed
	to monitor the environment (for example, the environment is an electromagnetic channel monitored with a photodetector
	\cite{salehteich}) and know which $M_j$ occurred. 
	Because in that case $\str \in\SOS$ we will use the 
	latter scenario in our model to derive some physically motivated operations on $\SOS$. As explained below
	some important quantum operations (amplitude damping, phase damping, etc.) have Kraus operators which preserve
	$\SOS$ and these operators can serve as measurement operators for our system.

	Considering the set of states as column vectors,
	$\{\ket{\psi_{\stm{\sr}}}=\left(\begin{smallmatrix}\sqrt{1-\sr}\\\sqrt{\sr}\end{smallmatrix}\right)\mid \stm{\sr}\in\SOS\}$ 
	is preserved by the real nonzero nonnegative-entry $2 \times 2$ matrices
	$M_2(\mathbb{R}^+)\setminus \{0\}$, up to a normalization factor. 

	A finite subset $\{K_j\} \subset M_2(\R^+)$ satisfying the completeness relation
	\begin{equation}
		\sum_{j}\adj{K_j}\noadj{K_j} = I_2, \label{measnorm}
	\end{equation}
	can represent the operators for either a complete measurement or a trace-preserving 
	quantum operation \cite{NC} (we think of it as the former).
	It is then a simple calculation to check that the $K_j$ belong to either or both the filtering class (rank 2)
		\[ \mathcal{F}=\{\big(\begin{smallmatrix}a && 0 \\ 0 && b\end{smallmatrix} \big), 
					\big(\begin{smallmatrix}0 && a \\ b && 0 \end{smallmatrix} \big) \mid
					ab \ne 0\}\]
		or the jump class (rank 1)
		\[ \mathcal{J}=\{\big(\begin{smallmatrix}0 && a \\ 0 && b \end{smallmatrix}\big),
					\big(\begin{smallmatrix}a && 0 \\ b && 0 \end{smallmatrix}\big) \}.
		\] 
		The operators in $\mathcal{F}$ arise from the Bayesian
		updating of the state conditioned on the result of the measurement while the operators
		in $\mathcal{J}$ are direct jumps (e.g. wave function collapse) to a fixed state. 

		An operator $K \in \mathcal{F} \cup \mathcal{J}$ is also considered a map on $\SOS$,
		\[ K:\stc\mapsto\str\]
		such that 
		\[\ket{\psi_{\str}} = \frac{K\ket{\psi_{\stc}}}{\norm{K\ket{\psi_{\stc}}}}. \]
		Furthermore if $K$ is taken to be a measurement operator then $\Pr(K|\stc)=\norm{K\ket{\psi_{\stc}}}^2$
		for all $\stc \in\SOS$.
		Table \ref{table:ops} describes the measurement operators in that sense. 
		
		Since the discrete
		systems in this paper are approximations of real continuous systems, in a future paper we
		investigate such operators in the infinitessimal time limit--the filtering operators go smoothly to
		the identity operator while the jump operators remain fixed jumps (although the probabilities of occurrence
		go to zero).
    \begin{table}\caption{Measurement operators}\label{table:ops}
    \begin{center}
	\begin{tabular}{rccl}\toprule
    Name & Operator $K$ & $s'$ & $\Pr(K|s)$ \\ \midrule
    $f_1(a,b)$ &
    \big($\begin{smallmatrix}\sqrt{a} && 0 \\ 0 && \sqrt{b} \end{smallmatrix}\big)$ & 
		$\frac{bs}{(b-a)s+a}$& $(b-a)s+a$ \\
	$f_2(a,b)$ &
	\big($\begin{smallmatrix}0 && \sqrt{b} \\ \sqrt{a} && 0 \end{smallmatrix}\big)$ &
		$\frac{a(1-s)}{(b-a)s + a}$ & $(b-a)s+a$ \\
	$j_1(a,b)$ &
	\big($\begin{smallmatrix}0 && \sqrt{a} \\ 0 && \sqrt{b} \end{smallmatrix}\big)$ &
		$\frac{b}{a+b}$ & $(a+b)s $ \\
	$j_2(a,b)$ &
	\big($\begin{smallmatrix}\sqrt{a} && 0 \\ \sqrt{b} && 0 \end{smallmatrix}\big)$ &
		$\frac{b}{a+b}$ & $(a+b)(1-s)$\\ \bottomrule
    \end{tabular}
    \end{center}
    \end{table}
    
    Some examples are \cite{NC} amplitude damping $\{f_1(1,1-\gamma),j_1(\gamma,0)\}$, 
	phase damping $\{f_1(1,1-\gamma),j_1(0,\gamma)\}$,
    generalized amplitude damping $\{f_1\big(p,p(1-\gamma)\big),f_1((1-p)(1-\gamma)),j_1(p\gamma,0),j_2(0,(1-p)\gamma)\}$ and
    the bit-flip channel $\{f_1(p,p),f_2(1-p,1-p)\}$.
\subsection{Filtering and Jump operators}
    When $M_i \in \mathcal{J}$ for all $0\le i < m$, the outcomes are in $m$ discrete states, call them $\ket{\psi_i}$,
	$0 \le i <m$, possibly nondistinct. There may be interesting problems when the set of the $\ket{\psi_i}$ are nonorthogonal,
	but we will not consider jump-only measurements further in this paper.

	The measurement sets $\mathcal{M}$ that we will use will contain at least one filtering operator and zero or more jump
	measurements so that $\mathcal{M}$ contains two or more measurement operators. Experience suggests that having more than
	two nonorthogonal measurement operators can make the system easier to control.
\section{Cost functions for a simple model}
	The two examples that we discuss in this paper differ only by their cost functions.
	For the first, the affine model, the terminal cost function is equation \eqref{eq:vnaff}, which,
	when recast in the $\sr$-parameterization is
	\begin{equation}\label{eq:jnaff}
		\tcfm{\text{aff}}(\sr_{\steps})=|\sr-\srt|=\sr,
		\end{equation}
	where the target state $\srt=0$ if $\stt=\ket{0}$.
	Other functions which have their global minima at the target states could be justified. 
	One might decide that a resultant state more than $0.2$ units from the target state in $\SOS$ is
	unacceptable. Then we might have a $\tcfm{}$ like
	\begin{equation}\label{eq:jnthresh}
		\tcfm{\text{thresh}}(\sr) = \begin{cases} |\sr-\srt|, & \text{if $|\sr-\srt| \le 0.2$} \\
							100, & \text{if $|\sr-\srt| > 0.2$.} \end{cases}
	\end{equation}
	The number $100$ was chosen to be much larger than typical values of the cost-to-go functions (essentially infinity)
	but not large enough to cause overflows in computation.
	
	With our discrete time formulation of the control problem, we do not directly deal with differential equations and
	so nonsmooth functions like $\tcfm{\text{thresh}}$ can be incorporated naturally into the model.
	\subsection{Cost-to-move function}
	In addition to the terminal cost there are resources that are used by the controller at every timestep to drive the
	system. The use of these resources is quantified by the cost-to-move function. For simplicity, as mentioned in 
	section \ref{mdp}, the cost-to-move function $\cost$ will only depend in the current state $\st$ and 
	the state $\stc$ that the controller will move the system to in the next timestep by using one of its Hamiltonians.
	
	Part of what makes quantum control hard is that the controller can not move the system to any state that it
	chooses in a single step. Otherwise, the step preparation problem would be trivial. The transitions
	$\st \mapsto \stc$ which are inadmissible (disallowed) in the model are due to physical considerations. 
	The Hamiltonians that the controller has available, the size of the timestep and the way the system is discretized determine which pairs $(\st,\stc)$ are inadmissible. In that case, we assign
	\[ \cost = \infty \]
	for an inadmissible pair $(\st,\stc)$. As in equation \eqref{eq:jnthresh}, 
	for computational purposes $\infty$ is taken to be a large number.
	The function $\ctgtipol{\ti}{\stc,*}$ is the cost-to-go if the next action is $\stc$ and then an optimal policy
	is followed thereafter. Optimizing $\ctgtipol{\ti}{\stc,*}$ over $\stc$ is the same as the optimization step in the 	
	dynamic programming algorithm (DPA) equation \eqref{dpa} and so it often shows up in implementations of the DPA.
	The advantage in making all pairs in $\SOS \times \SOS$ admissible and redefining infinity to a large number is that 
	it makes the domain of $\ctgtipol{\ti}{\src,*}(\sr)$
	(in variables $(\sr,\src)$) rectangular and finite on the entire domain, easing the optimization process.
	
	As for reasonable cost functions, the most straightforward metric, for states on the Bloch sphere $\bloch$, might
	be the angular distance.
	Two states $\sr,\src\in\SOS$ have angular distance
	\[ d_{\theta}(\sr,\src)=\arccos{\big(1-2\src-2\sr+4\sr\src+4\sqrt{\sr(1-\src)\src(1-\src)}\big)} \].
	
	One might imagine that in the $U_c$ control step above the controller has some Hamiltonians that can
	be switched on and off for a percentage of the time afforded to the controller. Then perhaps all states
	that are a certain distance away in the $d_{\theta}$ metric can be reached, for a similar use of resources.
	We might have then
	\begin{equation}\label{eq:gthresh}
		\costonly^\text{thresh}(\sr,\src) = \begin{cases} 0, & \text{if $d_{\theta}(\sr,\src) < 0.2$} \\
									100, & \text{if $d_{\theta}(\sr,\src)\ge 0.2$.} \end{cases}
	\end{equation}
	On a similar vein to the terminal cost, $100$ is an arbitrary `big' number so that if 
	$\ctgtipol{0}{*}(\sr) \ge 100$ then there is no way to control
	the system from $\sr$ to $\srt$ in the model.

	Other possible cost functions are $d_\theta$ or $d_\theta^2$. Here we imagine that applying
	the Hamiltonians in the control step $U_c$ cost energy which is proportional to how far in the
	$d_\theta$ metric that the state has moved \cite{margolus}. In this case we would like to balance the usage of energy
	to the closeness of the resultant state at the end of the experiment to the target state. In this situation,
	having infinite energy ($\cost=0$) might allow the controller to make very strong pulses toward
	the end of the experiment. Whether or not that is an optimal control solution depends on the probability
	of states near the target state decaying away before the end of the experiment. In \cite{belavkin04}
	the authors faced a similar problem.
	
	In order to make the system solvable analytically, we will use the second order Taylor polynomial of 
	$d_{\theta}^2$ 		expanded about $(\frac{1}{2},\frac{1}{2})$, giving a quadratic symmetric in its arguments.
	It is
	\begin{equation}\label{eq:gquad}
			\costonly^\text{aff}(\sr,\src) = 4(\sr-\src)^2. 
	\end{equation}
	Using the function  $C(s,\hat{s})=|s-\hat{s}|$ would be even simpler but leads to trivial controls when
	equation \eqref{dpa} is solved for the example model in the next section.
	%SOLUTIONS
	%4 Solutions num/anal
\section{Solving the DPA for the models}
	Putting the affine terminal cost \eqref{eq:jnaff} and the quadratic movement cost \eqref{eq:gquad} 
	into equation \eqref{dpa} we get
	\begin{align}\label{eq:dpaaff}
	\ctgtipol{\ti-1}{*}(\sr)& = 
			\min_{0\le \src\le 1}\left\{4(\sr-\src)^2 + 
				\sum_{j}{\ctgtipol{\ti}{*}\left(\frac{a_j\src +b_j}{c_j\src +d_j}\right)(c_j\src +d_j)}\right\},  \\
	\notag \tcf(\sr) &=\sr,
	\end{align}
	with constants $a_j$, $b_j$, $c_j$ and $d_j$ which depend on the measurement set chosen from table \ref{table:ops}.
	We will also 
	solve equation \eqref{dpa} numerically using the threshold cost functions \eqref{eq:jnthresh} and \eqref{eq:gthresh}.

	In order to solve equations \eqref{eq:dpaaff} numerically, we must specify the measurement set to be used.
	Numerous examples that we have experimented with have the same basic behavior, which is typified
		by the system with measurement operators
		$\{f_1(a,1),j_2(0,1-a)\}$. With some rescaling equation \eqref{eq:dpaaff} becomes
	\begin{equation}\label{eq:dpaaffa}
		\bellmanit{J}(\sr)=\min_{0\le \src\le1}\left\{
			(\sr-\src)^2+
			J\left(\frac{\src}{(1-a)\src+a}\right)\Big((1-a)\src+a\Big)+ \\
			J(1)\Big(1-\big((1-a)\src+a\big)\Big)
			\right\}.
	\end{equation}
	The parameter $0 \le a \le 1$ determines how unstable the target state is. If the system is a two-level atom
	then it is the rate of decay. We can vary $a$ in order to study how the optimal policies change with the instability
	of the system.
	\subsection{Solving numerically}
		A simple way to proceed numerically is to discretize $\SOS$ into say, 100 equally
		spaced points. For the resolution desired this naive assumption turns out to be acceptable
		for reasons that will become clear in the analytic solution section.
		The functions $\ctgtipol{\ti}{*}$ can be approximated as a discrete list of values and 
		computed recursively by minimizing over a discrete list. Figure \ref{graphs:all}a. shows a graph 
		of $\ctgtipol{\ti}{*}$ for $\ti=0,1,\ldots,5$ when $a=0.8$.
		Similar graphs are obtained for $0<a<1$ and for systems with other sets of measurement
		operators.
	\begin{figure}[!b]\label{graphs:all}
		\begin{center}
			\hspace{0.3em}
			\begin{overpic}[scale=.7]{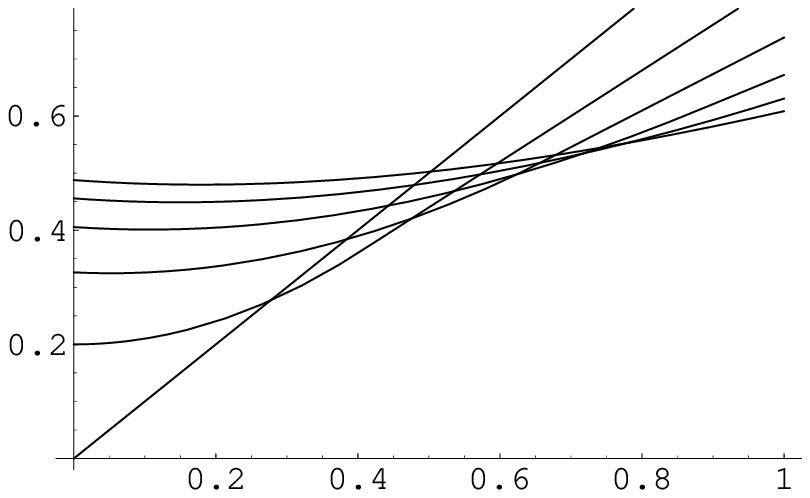} %J affine
				\put(0,60){a.}
				\put(100,4){\small$x$}	%Horizontal axis label
				\put(18,10){$\ctgtipol{5}{*\text{aff}}$}
				\put(10,45){$\ctgtipol{0}{*\text{aff}}$}
			\end{overpic}
		\quad\enspace
			\begin{overpic}[scale=.7]{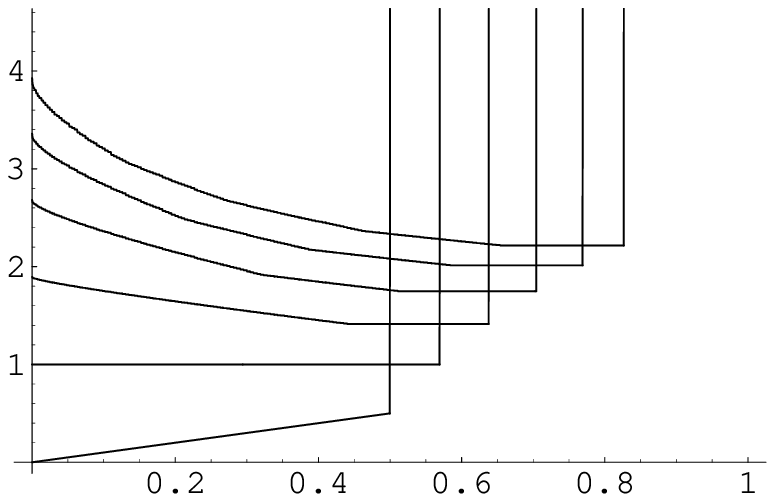} %J threshold
				\put(-3,60){b.}
				\put(98,4){\small$x$}
				\put(10,10){$\ctgtipol{5}{*\text{thresh}}$}
				\put(15,45){$\ctgtipol{0}{*\text{thresh}}$}
			\end{overpic}
			\\
	%row 2
			\begin{overpic}[scale=.7]{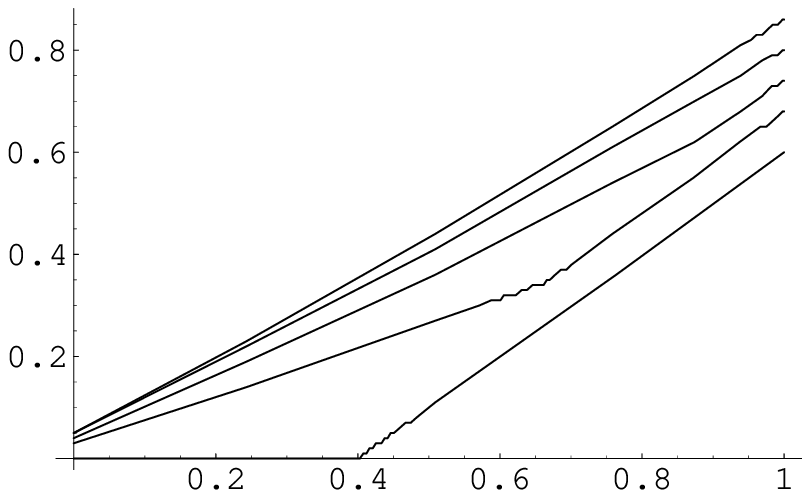}
				\put(0,60){c.}
				\put(100,4){\small$x$}
				\put(30,25){$\polf^{*\text{aff}}_0$}
				\put(55,10){$\polf^{*\text{aff}}_4$}
			\end{overpic}
		\quad
			\begin{overpic}[scale=.7]{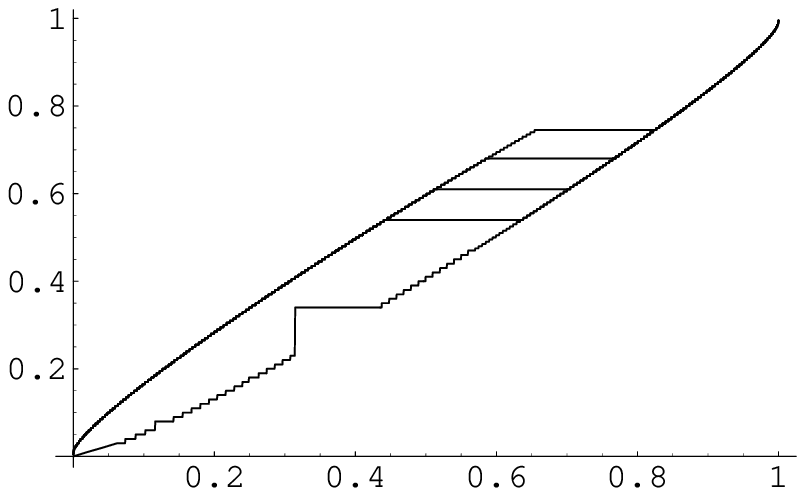}
				\put(0,60){d.}
				\put(99,4){\small$x$}
				\put(10,22){$\polf^{*\text{thresh}}_0$}
				\put(55,25){$\polf^{*\text{thresh}}_4$}
			\end{overpic}
			\caption{\small Example costs-to-go and optimal policies.}
		\label{fig-label}
		%
	%\begin{overpic}[scale=1,grid]
	%  {thresh.eps}
	%\end{overpic}
\end{center}
	\end{figure}
		Figure \ref{graphs:all}c. shows the optimal policies $\polo(\sr,\ti)$) for the affine problem.
		From the picture, the  $\polo(\sr,\ti)$ appear to be piecewise affine. The analytic
		solution will tell us that it is not the case, although the $\polo(\sr,\ti)$ are closely approximated
		by piecewise affine functions.
		
		When we use the threshold functions, a typical run might look like figure \ref{graphs:all}(b.,d.).
		The cost-to-go functions are quite different for the two models. One can easily see how the cutoff
		in $\tcf{\text{thresh}}$ propagates through the $\ctgtipol{\ti}{*}$. We expect an analytic analysis of this
		model to be fairly easy. Another interesting feature which is shared with the latter model is
		that the local minima for the various value functions do not lie at the target state. This is
		due to the large probability of a state near the target state decaying away before the end of the
		experiment.
	\subsection{Approximating the solution analytically}
		We can explain some of the behaviors that arose with the affine/quadratic model \eqref{eq:dpaaffa}.
		The goal is to show that if cost-to-go $J$ is quadratic then $\bellmanit{J}$ is approximately quadratic. 
		Appealing to the continuity of the iteration map, we can obtain reasonable approximations to 
		$\bellmanitpow{J}{j}$ for small values of $j$.
		Due to constant-function fixed points of $\bellmanit{\cdot}$,  for larger values of $j$, $\bellmanitpow{J}{j}$ 
		becomes uninteresting.
		
		Suppose $J(\sr)=A\sr^2+B\sr+C$. Considering the sum in equation \eqref{eq:dpaaff}, we get linear terms from
		jump-type measurements and for each filtering-type measurement operator we get a term like
		$J(\frac{a\src +b}{c\src +d})(c\src +d)$. Evaluating this, we get a quadratic term and a fractional term $N(\src)$
		given by
		\begin{equation}\label{nonlinterm}
			N(\src)=\frac{A(ad-bc)^2}{c^2(d+c\src)}.
		\end{equation}
	
		The second order Lagrange interpolating polynomial agreeing with $N(\src)$ at the
		points $\src_i=0,\frac{1}{2},1$ is
		\begin{equation*}
			L(\src)=
				\frac{A(ad-bc)^2}{c^2}\left(\frac{1}{d}-
				\frac{c(3c+2d)}{d(c+d)(c+2d)}\src+
				\frac{2c^2}{d(c+d)(c+2d)}\src^2\right)
			\end{equation*}
			with error bounds \cite{numerical}
			\begin{equation*}
				|N(\src)-L(\src)| \le \frac{\max_{0\le\xi\le 1}|N^{(3)}(\xi)|}{3!}|\src(\src-\frac{1}{2})(\src-1)|.
			\end{equation*}
			The maximum error is bounded by $\left|\frac{\alpha}{(d+c)^4 12\sqrt{3}}\right|$.
			For an operator $f_1(a,b)$ or $f_2(a,b)$, the bound becomes
			\[ \left|\frac{A a^2(b-a)}{b^2 12\sqrt{3}}\right|. \]
			For example, if $J$ was an affine function ($A=0$) then the approximation is exact.
	
			Let the approximate to $\bellmanit{J}$ be called $\bellmanitapp{J}$. 
			Computing $\bellmanitapp{J}(\sr)$ amounts to finding the minimum of a quadratic parameterized 
			by $\sr$ on a closed interval. 
			$\bellmanitapp{J}$ is thus a piecewise quadratic function. 
			Likewise to approximate $\bellmanitpow{J}{2}$ we compute $\bellmanitapp{\bellmanitapp{J}}$
			and get a bound on the error.
			We do not find the occasion to repeat this procedure to find, for example an approximation to 
			$\ctgtipol{0}{*}$ due to the algebra involved and the efficacy of the numerical method. 
			The analytic solution explains why the graphs of
			the $\polo$ look piecewise linear -- because the graphs of the $\ctgtipol{\ti}{*}$ look piecewise quadratic.
%CONCLUSION
\section{Conclusion}
	We have demonstrated a method for controlling a quantum system by controlling its simulation.
	We have shown that the method works by applying it to two simple models. The method has been
	applied to more complicated models: a qubit undergoing decoherence, which has a more complicated
	state space, and two qubits undergoing disentanglement. These situations will appear in a future
	paper.
	
	We need to resolve issues with the discretization. How small does the timeslice have to be
	to get `good' control? How does the MDP turn into a stochastic differential equation as the
	timeslice goes to $0$?

	We are concerned with how the control problem scales with the size of the system. When searching
	the state space in the optimization step we rely on good parameterizations of the system. These
	become harder to come by as the size of the system increases.
\begin{acknowledgments}
	This work has been partially supported by the National Science Foundation under
	grant ECS-0202087 and by the Defense Advanced Research Projects Agency as part of the
	Quantum Entanglement Science and Technology program under grant N66001-09-1-2025.
\end{acknowledgments}
\bibliography{qcontrol1qb}
\end{document}